\documentclass[runningheads]{llncs}

\usepackage{graphicx}

\usepackage{algorithm}
\usepackage{algpseudocode}

\usepackage{hyperref}

\usepackage{amssymb}
\usepackage{amsmath}
\usepackage{comment}

\newtheorem{observation}{Observation}

\newcommand{\OPT}{\operatorname{OPT}}
\newcommand{\ALG}{\operatorname{ALG}}
\newcommand{\GRD}{\operatorname{GRD}}

\usepackage{xcolor}
\newcommand{\inbal}[1]{{\color{red}{[ITC: #1]}}}

\newcommand{\ameer}[1]{{\color{blue}{[AA: #1]}}}

\begin{document}

\title{Auctions with Interdependence and SOS: Improved Approximation}

\titlerunning{Auctions with Interdependence and SOS}

\author{Ameer Amer \and
Inbal Talgam-Cohen}

\authorrunning{A.~Amer and I.~Talgam-Cohen}

\institute{Faculty of Computer Science, Technion -- Israel Institute of Technology,
\email{\{ameeramer,italgam\}@cs.technion.ac.il}}

\maketitle

\begin{abstract}
Interdependent values make basic auction design tasks -- in particular maximizing welfare truthfully in single-item auctions -- quite challenging. Eden et al.~recently established that if bidders’ valuation functions are submodular over their signals (a.k.a.~SOS), a truthful 4-approximation to the optimal welfare exists. We show existence of a mechanism that is truthful and achieves a tight 2-approximation to the optimal welfare when signals are binary. Our mechanism is randomized and assigns bidders only $0$ or $\frac{1}{2}$ probabilities of winning the item. Our results utilize properties of submodular set functions, and extend to matroid settings. 

\end{abstract}

\section{Introduction}

One of the greatest contributions of Robert Wilson and Paul Milgrom, the 2020 Nobel Laureates in economics, is their formulation of a framework for auction design with \emph{interdependent} values \cite{cite-scientific-background-for-the-noble-prize}. Up to their work, the standard assumption underlying auction design theory was that each bidder fully knows her value for the item being auctioned, because this value depends only on her own private information. This assumption is, however, far from reality in very important settings -- for example, when the auction is for drilling rights, the information one bidder has about whether or not there is oil to be found is extremely relevant to how another bidder evaluates the rights being auctioned. Works like \cite{Wilson77-A-bidding-model-of-perfect-competition} and \cite{MilgromW82} lay the foundation for rigorous mathematical research of such settings, yet many key questions still remain unanswered.

For concreteness, consider an auction with a single item for sale (our main setting of interest).
In the interdependent values model, every bidder $i\in[n]$ has a privately-known signal $s_i$, and her value $v_i$ is a (publicly-known) function of all the signals, i.e., $v_i=v_i(s_1,s_2,...,s_n)$. Thus, in this model, not only the auctioneer is in the dark regarding a bidder's willingness to pay for the item being auctioned; so is the bidder herself (who knows $s_i$ and $v_i(\cdot)$ but not $s_{-i}$)! 

This stark difference from the standard, independent private values (IPV) model creates a big gap in our ability to perform seemingly-simple auction design tasks. Arguably the most fundamental such task is \emph{truthful welfare maximization}. For IPV, the truthful welfare-maximizing Vickrey auction \cite{Vickrey61-counterspeculation-paper} is a pillar of mechanism design (e.g., it has many practical applications and is usually the first auction taught in a mechanism design course). But with interdependence, welfare and truthfulness are no longer perfectly compatible: Consider two bidders reporting their signals $s_1,s_2$ to the auction, which allocates the item to the highest-value bidder according to these reports; if the valuation functions $v_1,v_2$ are such that bidder 1 wins when $s_1=0$ but loses when $s_1=1$, this natural generalization of Vickrey to interdependence is non-monotone and thus non-truthful. This is the case, for example, if $v_1=1+s_1$ and $v_2=H\cdot s_1$ for $H>2$ (see \cite[Example 1.2]{EdenFFG19}). 

The classic economics literature addressed this challenge by introducing a somewhat stringent condition on the valuation functions called ``single-crossing'', which ensures truthfulness of the natural generalization of Vickrey (in particular, single-crossing is violated by $v_1=1+s_1,v_2=H\cdot s_1$). Recently, a breakthrough result of Eden et al.~\cite{EdenFFG19} took a different approach: For simplicity consider \emph{binary} signals -- e.g., ``oil'' or ``no oil'' in an auction for drilling rights. Formally, $s_i\in\{0,1\}$ (we focus on the binary case throughout the paper). The valuations are now simply \emph{set functions} over the signals, objects for which a rich mathematical theory exists. Eden et al.~applied a \emph{submodularity} assumption to these set functions (in particular, submodularity holds for $v_1=1+s_1$ and $v_2=H\cdot s_1$). 
Under such submodularity over the signals (\emph{SOS}), they shifted focus from maximizing welfare to \emph{approximating} the optimal welfare. While they showed that no truthful mechanism can achieve a better approximation factor than 2 (guaranteeing more than half the optimal welfare), they constructed a truthful randomized mechanism that achieves a 4-approximation (guaranteeing at least a quarter of the optimal welfare). The gap between $2$ and $4$ was left as an open problem.

\subsubsection{Our Results and Organization.}

In this work we resolve the above open problem of~\cite{EdenFFG19} for binary signals. 
More precisely, we show that in the binary signal case there exists a truthful randomized mechanism that achieves a 2-approximation to the optimal welfare (for a formal statement see Theorem~\ref{thm:main}). Our result holds for any number $n$ of bidders, and is constructive -- that is, we give an algorithm that gets the $n$ valuation functions as input, and returns the mechanism as output.%
\footnote{The algorithm runs in time polynomial in its input size, which consists of set functions over $n$ elements and so is exponential in $n$.}

The fact that our mechanism is randomized is unsurprising given another result of Eden et al.~\cite{EdenFFG19}, who show that a deterministic mechanism cannot achieve a constant approximation to the optimal welfare even with SOS. This result is in fact proved with the above example of $v_1=1+s_1, v_2=H\cdot s_1$ and $s_i\in\{0,1\}$. An interesting corollary of our construction is that a $2$-approximation is achievable by a mechanism that is only ``slightly'' randomized -- the  only allocation probabilities it uses are $0$ and~$\frac{1}{2}$.

Our algorithm is arguably quite simple and streamlined -- for every signal profile it searches for a feasible pair of bidders whose aggregate value exceeds that of the highest bidder, and randomly allocates the item among these two (this explains the factor of 2 in the approximation guarantee). Only if no such pair exists, the item is randomly either allocated to the highest bidder or left unallocated. To maintain monotonicity, the algorithm \emph{propagates} allocation probabilities to neighboring signal profiles. Despite its relative simplicity, the algorithm requires careful analysis, which in particular relies on new properties of collections of submodular functions (Section~\ref{sub:main-lemmas}). The main technical challenge is in showing that the 2-approximation guarantee holds despite the propagations.

{\it Example.}
To illustrate our method, consider again the above example of $v_1=1+s_1$ and $v_2=H\cdot s_1$ where $s_i\in\{0,1\}$. Our algorithm returns a randomized allocation rule that gives the item to bidder~$1$ with probability~$\frac{1}{2}$ if $s_1=0$, and randomly allocates it to one of the two bidders if $s_1=1$.%
\footnote{Our algorithm has two iterations: At $s_1=0$, an appropriate pair is not found and so the highest bidder (bidder $1$) wins the item with probability~$\frac{1}{2}$, which is propagated forward to this bidder at $s_1=1$. At $s_1=1$, an appropriate pair is again not found and so the highest bidder (bidder $2$) wins the item with probability~$\frac{1}{2}$.}
This allocation rule is monotone (unlike the natural generalization of Vickrey), and leads to a truthful mechanism with a 2-approximation guarantee.

{\it Extensions.} 
In Appendix~\ref{matroid} we extend our main result to beyond single-item settings, namely to general single-parameter settings in which the set of winning bidders must satisfy a \emph{matroid} constraint~\cite{Oxley's-book-on-matroids}. 
As in \cite{EdenFFG19}, we can also extend our positive results 
from welfare to \emph{revenue} maximization using a reduction of \cite{EdenFFG18}.

{\it Organization.} 
After presenting the preliminaries in Section~\ref{sec:preliminaries}, we state our main theorem and give an overview of our algorithm in Section~\ref{sec:main-result}. The analysis appears in Section~\ref{sec:analysis} and Appendix~\ref{appx:missing}. Section~\ref{sec:summary} summarizes with future directions. Appendix~\ref{appx:pseudocode} includes the algorithm and running time, Appendix~\ref{matroid} the extension to matroids and Appendix~\ref{over-binary} our results for non-binary signals.

\subsubsection{Additional Related Work.}

Interdependent values have been extensively studied in the economic literature (see, e.g., \cite{DasguptaM00,efficient-inter,Maskin92,Ausubel99}).
In computer science, most works to date focus on the objective of maximizing revenue \cite{ChawlaFK14,RoughgardenT16,Li17,EdenFFG18}.
The work of \cite{EdenFFG18} considers welfare maximization with a relaxed \emph{$c$-single-crossing} assumption, where parameter $c\geq1$ measures how close the valuations are to satisfying classic single-crossing. This work achieves a 
$c$-approximation for settings with binary signals. 
Their mechanisms also use propagations but otherwise are quite different than ours. 
{{The work of \cite{Vasilis2021} 
also focuses on welfare but does not assume single-crossing; instead it partitions the bidders into $\ell$ ``expertise groups'' based on how their signal can impact the values for the good, and using clock auctions achieves approximation results parameterized by $\ell$ (and by the number of possible signals).}} 
The main paper our work is inspired by is \cite{EdenFFG19}. It introduces the \emph{Random Sampling Vickrey auction}, which by excluding roughly half of the bidders achieves a $4$-approximation to the optimal welfare for single-parameter, downward-closed SOS environments.
The authors also show positive results for combinatorial SOS environments under various natural constraints. 
%
Finally, \cite{EdenFTZ21} also study welfare maximization in single- and multi-parameter environments but by simple, non-truthful parallel auctions.

\section{Setting}
\label{sec:preliminaries}

{\it Signals.} As our main setting of interest we consider single-item auctions with $n$ bidders. Every bidder $i$ has a binary \emph{signal} $s_i\in\{0,1\}$, which encompasses the bidder's private bit of information about the item.
It is convenient to identify a signal profile $s=(s_1,s_2,...,s_n)$ with its corresponding set $S=\{i \mid s_i=1\}$ (by treating $s_i$ as an indicator of whether $i\in S$).

{\it Values.} The bidders have \emph{interdependent values} for the item being auctioned: Every bidder $i$'s value $v_i$ is a non-negative function of all bidders' signals, i.e., $v_i=v_i(s_1,s_2,...,s_n)\ge 0$. We adopt the standard assumption that the valuation function $v_i$ is weakly increasing in each coordinate and strongly increasing in~$s_i$. Using the set notation we also write $v_i=v_i(S)$.%
\footnote{This notation is not to be confused with the value for a \emph{set of items} $S$; in our model there is a \emph{single} item, and a bidder's interdependent value for it is determined by the \emph{set of signals}, i.e., which subset of signals is ``on''.}
This makes $v_i(\cdot)$ a monotone \emph{set function} over subsets of~$[n]$.

{\it Who knows what.} A setting is summarized by the valuation functions $v_1,\dots,v_n$, which are publicly known (as is the signal domain $\{0,1\}$). The instantiation of the signals is private knowledge, that is, signal $s_i$ is known only to bidder~$i$.  

\subsubsection{SOS Valuations.}

The term \emph{SOS valuations} was coined by Eden et al.~\cite{EdenFFG19} to describe interdependent valuation functions that are \emph{submodular over the signals} (see also \cite{ChawlaFK14,EdenFFG18,NiazadehRW20}).%
\footnote{As mentioned above, submodularity over signals is not to be confused with submodularity over items in combinatorial auctions.}
With binary signals, valuations are SOS if $v_i(\cdot)$ is a submodular set function for every $i\in[n]$.


\begin{definition}[Submodular set function]
A set function $v_i:2^{[n]}\rightarrow \mathbb{R}$ is submodular if for every $S,T\subseteq [n]$ such that $S\subseteq T$ and $i\in [n]\setminus T$ it holds that $v_i(S\cup \{i\})-v_i(S)\geq v_i(T\cup \{i\})-v_i(T)$.
\end{definition}

A weaker definition that will also be useful for us is subadditivity. Every submodular set function is subadditive, but not vice versa.

\begin{definition}[Subadditive set function]
A set function $v_i:2^{[n]}\rightarrow \mathbb{R}$ is subadditive if for every $S,T\subseteq [n]$ is holds that $f(S)+f(T)\geq f(S\cup T)$.
\end{definition}

Given a set function $v_i$ and a subset $S\subseteq[n]$, we use $v_i(\cdot\mid S)$ to denote the following set function: $v_i(T\mid S)=v_i(T\cup S)-v_i(S)$ for every $T\subseteq[n]$. It is known that submodularity of $v_i$ implies subadditivity of $v_i(\cdot\mid S)$:

\begin{proposition}[e.g., Lemma 1 of \cite{LehmannLN06}]
\label{pro:subadd}
If $v_i$ is submodular then $v_i(\cdot\mid S)$ is subadditive for every subset $S\subseteq[n]$. 
\end{proposition}

We refer to an ordering $S_1,S_2,\dots,S_{2^n}$ of all subsets of the ground set of elements~$[n]$ as \emph{inclusion-compatible} if for every pair $S_k,S_\ell$ such that $S_k\subset S_{\ell}$, it holds that $k<\ell$ (the included set is before the including one in the ordering).

{\it Notation.} Consider a set function $v_i$ and two elements $j,k\in[n]$. For brevity we often write $v_i(j)$ for $v_i(\{j\})$, and $v_i(jk)$ for $v_i(\{j,k\})$.




\subsection{Auctions with Interdependence}
\label{sub:auctions-interdep}

{\it Randomized mechanisms.} Due to strong impossibility results for deterministic mechanisms \cite{EdenFFG19}, we focus on randomized mechanisms as follows: 
A randomized mechanism $M=(x,p)$ for interdependent values is a pair of allocation rule $x$ and payment rule $p$. The mechanism solicits signal reports from the bidders, and maps a reported signal profile $s$ to non-negative allocations $x=(x_1,...,x_n)$ and expected payments $p=(p_1,...,p_n)$, such that the item is \emph{feasibly} allocated to bidder $i$ with probability $x_i$ (feasibility means $\sum_{i=1}^{n}{x_i(s)}\leq 1$).

{\it Truthfulness.} With interdependence, it is well-established that the appropriate notion of truthfulness is \emph{ex post IC} (incentive compatibility) and \emph{IR} (individual rationality). Mechanism $M$ is ex post IC-IR if the following holds for every bidder $i$, true signal profile~$s$ and reported signal~$s'_i$: Consider bidder $i$'s expected utility when the others truthfully report $s_{-i}$: 
$$
x_i(s_{-i},s'_i)v_i(s)-p_i(s_{-i},s'_i);
$$  
then this expected utility is non-negative and maximized by truthfully reporting $s'_i=s_i$.%
\footnote{Note the difference from \emph{dominant-strategy} IC, in which this guarantee should hold no matter how other bidders report.}

Similarly to independent private values, the literature on interdependent values provides a characterization of ex post IC-IR mechanisms -- as the class of mechanisms with a \emph{monotone} allocation rule $x$. Allocation rule $x$ satisfies monotonicity if for every signal profile $s$, bidder $i$ and $\delta\geq 0$, increasing $i$'s signal report by $\delta$ while holding other signals fixed increases $i$'s allocation probability:
$$
x_i(s_{-i},s_i)\leq x_i(s_{-i},s_i+\delta).
$$
The characterization also gives a payment formula which, coupled with the monotone allocation rule, results in an ex post IC-IR mechanism. In more detail, the expected payment of bidder $i$ is achieved by finding her critical signal report and plugging it into her valuation function while holding others' signals fixed (see \cite{RoughgardenT16} for a comprehensive derivation of the payments).

{\it Welfare maximization.} Our objective in this work is to design ex post IC-IR mechanisms for interdependent values that maximize social welfare. 
For a given setting and true signal profile $s$, the optimal welfare $\OPT(s)$ is achieved by giving the item to the bidder with the highest value, i.e., 
$\OPT(s)=\max_i \{v_i(s)\}$.
Given a randomized ex post IC-IR mechanism $M=(x,p)$ for this setting, 
$\ALG(s)$ is its welfare in expectation over the internal randomness, i.e., 
$\ALG(s)=\sum_{i=1}^{n}{x_i(s) v_i(s)}.$
We say mechanism $M$ achieves a $c$-approximation to the optimal welfare for a given setting if for every signal profile $s$, $\ALG(s)\geq \frac{1}{c} \OPT(s)$ (note that the required approximation guarantee here is ``universal'', i.e., should hold individually for every~$s$). Since Eden et al.~\cite{EdenFFG19} devise a setting for which no randomized ex post IC-IR mechanism can achieve better than a $2$-approximation, we aim to design mechanisms that achieve a $c$-approximation to the optimal welfare where $c\ge 2$ (the closer to $2$ the better).

\section{Main Result and Construction Overview}
\label{sec:main-result}

Our main result is the following:

\begin{theorem}
    \label{thm:main}
    For every single-item auction setting with $n$ bidders, binary signals and interdependent SOS valuations, there exists an ex post IC-IR mechanism that achieves a $2$-approximation to the optimal welfare.
\end{theorem}

Our proof of Theorem~\ref{thm:main} is constructive -- we design an algorithm (Algorithm~\ref{alg:main}) that gets as input the valuation functions $v_1,\dots,v_n$, and outputs an allocation rule $x$.
Note that the main goal of the algorithm is to establish existence.
Rule $x$ is guaranteed to be both feasible and monotone. Thus, coupled with the appropriate expected payments $p$ (based on critical signal reports), it constitutes an ex post IC-IR mechanism $M=(x,p)$.
The main technical challenge is in showing that mechanism $M$ has the following welfare guarantee: for every signal profile $s$, $\ALG(s)\ge \frac{1}{2}\OPT(s)$. We prove this approximation ratio and establish $x$'s other properties like monotonicity in Section~\ref{sec:analysis}. The algorithm itself appears in Appendix~\ref{appx:pseudocode}; we now give an overview of its construction.


\subsection{Construction Overview}
\label{sub:construct-overview}

In this section we give an overview of Algorithm~\ref{alg:main}.
The algorithm maintains an ``allocation table'' with rows corresponding to the $n$ bidders, and columns corresponding to subsets $S\subseteq [n]$. 
At termination, column $S$ will represent the allocation rule $x(S)$, with entry $(i,S)$ encoding $x_i(S)$. For clarity of presentation the encoding is via colors: 
At initiation, all entries of the table are colored {white} to indicate they have not yet been processed.
During its run, the algorithm colors each entry $(i,S)$ of the table either {red} or {black}. Once a cell has been colored red or black, its color remains invariant until termination. The colors represent allocation probabilities as follows:
\begin{itemize}
    \item {red} = bidder $i$ gets the item with probability $\frac{1}{2}$ at signal profile $S$;
    \item {black} = $i$ does not get the item at $S$ (i.e., gets the item with probability~$0$).
\end{itemize}
As an interesting consequence, the allocation rule achieving the 2-approximation guarantee of Theorem~\ref{thm:main} uses only two allocation probabilities, namely $\frac{1}{2}$ and $0$. Note that for feasibility, no more than two entries in a single column should be colored {red}, and the remaining entries should be colored {black}.

We now explain roughly how the colors are determined by the algorithm.
Consider an inclusion-compatible ordering of all subsets of $[n]$. 
The algorithm iterates over the ordered subsets, with $S$ representing the current subset. We say bidder $i$ \emph{can be colored {red} at iteration $S$} if at the beginning of the iteration, $(i,S)$ is colored either {white} or {red}.
We define a notion of \emph{favored bidder(s) at iteration $S$} -- these are the ones the algorithm ``favors'' as winners of the item given signal profile $S$, and so will color them {red} at $S$. First,
if there is a pair of bidders $i\ne j$ for which the following conditions all hold, we say they are favored at iteration $S$ with \textbf{Priority 1}: 
\begin{enumerate}
    \item Bidder $i$ and $j$'s signals both belong to $S$ (i.e., $s_i=s_j=1$); 
    \item Bidders $i$ and $j$ can both be colored {red} at iteration $S$;
    \item No other bidder $k\neq i,j$ is colored {red} at the beginning of iteration $S$;
    \item The sum of values $v_i(S)+v_j(S)$ is at least $\OPT(S)$ (recall that $\OPT(S)$ is the highest value of any bidder for the item given signal profile $S$).
\end{enumerate}
If such a pair does not exist, but there exists a bidder $i$ who satisfies the following alternative conditions, we say $i$ is favored at iteration $S$ with \textbf{Priority 2}:
\begin{enumerate}
    \item Bidder $i$ can be colored {red} at iteration $S$;
    \item The value $v_i(S)$ equals $\OPT(S)$.
\end{enumerate}

Our main technical result in the analysis of the algorithm is to show that, unless at the beginning of iteration~$S$ two bidders are already colored {red}, then 
one of the two cases above must hold. That is, in every iteration~$S$ with no two reds, there is always either a favored pair with \textbf{Priority 1}, or a single favored bidder with \textbf{Priority 2}. 
%
Assuming this holds, the algorithm proceeds as follows. At iteration $S$ it checks whether two bidders are already {red}, and if so continues to the next iteration. Otherwise, it colors the favored bidder(s) {red} by priority, and all other bidders {black}. The algorithm then performs \emph{propagation} to other subsets $S'$ in order to maintain monotonicity of the allocation rule (the term propagation was introduced in our context by \cite{EdenFFG19}): 
\begin{itemize}
    \item If bidder $i\notin S$ is colored {red} at subset $S$, then {red} is propagated \emph{forward} to bidder~$i$ at subset $S'=S\cup\{i\}$.
    \item If bidder $i\in S$ is colored {black} at subset $S$, then {black} is propagated \emph{backward} to bidder~$i$ at subset $S'=S\setminus\{i\}$.
\end{itemize}
This completes the overview of our construction.

\section{Proof of Theorem~\ref{thm:main}}
\label{sec:analysis}

We begin with a simple but useful observation:

\begin{observation}
\label{obs:low-red-is-highest}
Consider a subset $S\subseteq[n]$ and $i\notin S$. If during iteration $S$ bidder $i$ is colored red then $v_i(S)=\OPT(S)$. 
\end{observation}

\begin{proof}
The algorithm colors a bidder with a low signal red only if this bidder has {\bf Priority 2}, and in this case her value must be highest among all bidders. 
\end{proof}


\noindent We now prove our main theorem, up to three lemmas which we prove in Sections~\ref{sub:no-errors}-\ref{sub:2-approx}, respectively. Section~\ref{sub:main-lemmas} also develops a necessary tool for the proof in Section~\ref{sub:2-approx}.

\begin{proof}[Theorem~\ref{thm:main}]
We show that Algorithm~\ref{alg:main} returns an allocation rule that is feasible, monotone, and achieves a $2$-approximation to the optimal welfare. For such an allocation rule there exist payments that result in an ex post IC-IR mechanism (see Section~\ref{sub:auctions-interdep}), establishing the theorem.

Let $x$ be the allocation rule returned by Algorithm~\ref{alg:main}.
We first show $x$ is \emph{feasible}. That is, for every subset $S\subseteq[b]$, the algorithm colors $(i,S)$ either red or black for every bidder $i$, and at most two bidders are colored red in column $S$.
To show this we invoke Lemma~\ref{lem:no-errors} below, by which Algorithm~\ref{alg:main} never reaches one of its error lines. 
Given that there are no errors, observe that the algorithm goes over all subsets, and for every subset $S\subseteq[n]$ either (i)~skips to the next subset (if two bidders are already red), or (ii)~finds a {\bf Priority~1} pair or {\bf Priority~2} bidder and colors them red. Indeed, by Lemma~\ref{lem:2-approx}, if (i)~does not occur then (ii)~is necessarily successful. Once a {\bf Priority~1} pair or {\bf Priority~2} bidder is found, the rest of the column is colored black. Furthermore, once any two bidders in a column are colored red, the rest of the column is colored black. This establishes feasibility.

We now show $x$ is \emph{monotone}. 
Since the only allocation probabilities $x$ assigns are $\frac{1}{2}$ and $0$ (and one of these is always assigned), it is sufficient to show that for every $S\subseteq[n]$ and $i\notin S$, if $x(i,S)=\frac{1}{2}$ then $x(i,S\cup\{i\})=\frac{1}{2}$.
This holds since every time the algorithm calls \textsc{ColorRed} to color $(i,S)$, it propagates the color red forward to $(i,S\cup\{i\})$ as well. 

It remains to show that $x$ achieves a $2$-approximation to the optimal welfare. By definition of {\bf Priority~1} and {\bf Priority~2}, if such bidders are colored red then a $2$-approximation is achieved for the corresponding signal profiles. It remains to consider signal subsets $S$ for which at the beginning of iteration $S$, two cells $i,j$ in the column are already colored red. These reds propagated forward from $v_i(S\setminus\{i\})$ and $v_j(S\setminus\{j\})$. 
Let $v_k(S)$ be the highest value at $S$. 
By Observation~\ref{obs:low-red-is-highest}, $v_i(S\setminus\{i\})$ and $v_j(S\setminus\{j\})$ are highest at $S\setminus\{i\}$ and $S\setminus\{j\}$, respectively: 
    \begin{align*}
    v_i(S\setminus\{i\}) & \geq v_k(S\setminus\{i\});\\
    v_j(S\setminus\{j\}) & \geq v_k(S\setminus\{j\}).
    \end{align*}
Applying Lemma~\ref{lemma1} to the above inequalities, it cannot simultaneously hold that $v_k(S)>v_i(S)+v_j(S)$.
So $v_i(S)+v_j(S)\ge v_k(S)$, and the approximation guarantee holds, completing the proof. 
\qed
\end{proof}

\subsection{No errors}
\label{sub:no-errors}

In this section we show that the algorithm runs without producing an error.

\begin{lemma}[No errors]
\label{lem:no-errors}
    Algorithm~\ref{alg:main} never reaches error lines 35, 41 or~55. 
\end{lemma}

\begin{proof}
We first establish that \textsc{ColorRed} is only called on white or red cells (so line 35 is never reached).
In Algorithm~\ref{alg:main}'s main procedure \textsc{Allocate}, sub-procedure \textsc{ColorRed} is called on cells that correspond to bidder $b$ of {\bf Priority~1} or {\bf Priority~2} (see lines 16-17 and 25), i.e., it is called only after verifying $b$ is not colored {black} (lines 13 and 23). Sub-procedure \textsc{ColorRed} also calls itself (line 50), after being called on $S$ from \textsc{Allocate}; in this case the color red is being propagated forward from $S$ to $S\cup\{b\}$. Assume for contradiction $(b,S\cup\{b\})$ is already colored black. This means \textsc{ColorBlack} was called on $(b,S\cup\{b\})$. But then due to backward propagation, $(b,S)$ would already be black too, in contradiction to \textsc{ColorRed} being called on it from \textsc{Allocate}.

We now show \textsc{ColorBlack} is only called on white or {black} cells (so line 55 is never reached).
Sub-procedure \textsc{ColorBlack} is called from \textsc{Allocate} (line 27) inside a condition verifying the current color is white.
\textsc{ColorBlack} is also called from \textsc{ColorRed} at (line 46); in this case the current subset $S$ already has two red bidders, on which \textsc{ColorBlack} is \emph{not} called, and there can never be more than two reds due to the condition at line 41 of \textsc{ColorRed}. 
Finally, \textsc{ColorBlack} also calls itself (line 62), after being called on $S$ from \textsc{Allocate} or from \textsc{ColorRed}; in this case the color black is being propagated backward from $S\cup\{b\}$ to $S$. Assume for contradiction $(b,S)$ is already colored red. This means \textsc{ColorRed} was called on $(b,S)$. But then due to forward propagation, $(b,S\cup\{b\})$ would already be red too, in contradiction to \textsc{ColorBlack} being called on it from \textsc{Allocate} or from \textsc{ColorRed}.

It remains to show that Algorithm~\ref{alg:main} never reaches error line 41, i.e., never attempts to color more than two bidders red for a given subset $S$. Note that such an attempt will not be made from \textsc{Allocate} due to the conditions at lines 8 and 14. Assume for contradiction it happens because of forward propagation from $S$ to $S\cup\{b\}$. Then there are two bidders $b_1\ne b_2$ distinct from $b$ who are already colored red for signal profile $S\cup\{b\}$ at the beginning of iteration $S$. This means that all other bidders including $b$ are colored black for signal profile $S\cup\{b\}$ at the beginning of iteration $S$ (see lines 44-48, in which after every time a bidder is colored red the algorithm checks if there are now two red bidders and if so colors all other bidders {black}). But then $(b,S\cup\{b\})$ is black when Algorithm~\ref{alg:main} attempts to color it red during forward propagation, contradiction.
\qed
\end{proof}

\subsection{Properties of SOS Valuations}
\label{sub:main-lemmas}

In this section we state and prove two lemmas for SOS valuations. The first is used in the proof of Theorem~\ref{thm:main}, and the second is the workhorse driving the proof in Section~\ref{sub:2-approx} that either {\bf Priority~1} or {\bf Priority~2} always hold. Very roughly, the first lemma (Lemma~\ref{lemma1}) states that if bidders $i,j$ have higher values than bidder $k$ when their own signals are low, then bidder $k$'s value cannot exceed their sum when their signals are high.
The second lemma (Lemma~\ref{lemma2}) is more complex, and to give intuition for what it states we provide a visualization in Figures~\ref{fig:block}-\ref{fig:case_full} (we use a similar visualization to sketch our main proof in Section~\ref{sub:2-approx}). The proof is by induction and is deferred to Appendix~\ref{appx:lemma2-pf}. 

\begin{lemma}
\label{lemma1}
Consider a subset $S'\subseteq [n]$ and three bidders $i,j,k$ (not necessarily distinct) with SOS valuations over binary signals; let $S^*=S'\cup\{i,j\}$. If the following three inequalities hold simultaneously then they all hold with equality:
\begin{itemize}
    \item $v_i(S^*\setminus\{i\})\geq v_k(S^*\setminus\{i\})$;
    \item $v_j(S^*\setminus\{j\})\geq v_k(S^*\setminus\{j\})$;
    \item $v_k(S^*)\ge v_i(S^*)+v_j(S^*)$.
\end{itemize}
\end{lemma}

\begin{proof}
Define $u_i(\cdot)=v_i(\cdot\mid S')$ for every $i\in [n]$. Using this notation, to prove the lemma we need to show that if the following three inequalities hold simultaneously, they must all hold with equality:
\begin{itemize}
    \item $u_i(j)+v_i(S')\geq u_k(j)+v_k(S')$;
    \item $u_j(i)+v_j(S')\geq u_k(i)+v_k(S')$;
    \item $u_k(ij)+v_k(S')\ge u_i(ij)+v_i(S')+u_j(ij)+v_j(S')$.
\end{itemize}
Assume the inequalities hold. Summing them and simplifying we get
\begin{eqnarray}
u_i(j)+u_j(i)+u_k(ij)\ge  u_k(j)+u_k(i)+v_k(S')+u_i(ij)+u_j(ij).\label{eq:sum}
\end{eqnarray}
We now use the fact that $u_k(\cdot)$ is subadditive (Proposition~\ref{pro:subadd}), so $u_k(ij)\le  u_k(j)+u_k(i)$. Thus by Inequality~\eqref{eq:sum},
$$
u_i(j)+u_j(i)\ge v_k(S')+u_i(ij)+u_j(ij).
$$
By monotonicity of set functions $u_i$ and $u_j$, $u_i(j)\leq u_i(ij)$ and $u_j(i)\leq u_j(ij)$. We conclude that $v_k(S')\le 0$, which can hold only with equality. But this equality would be violated if one of the three inequalities was strict, completing the proof.\qed
\end{proof}

\begin{lemma}
\label{lemma2}
Consider a subset $S'\subseteq [n]$ and $3+\ell_1+\ell_2+\ell_3$ bidders $E=\{i,j,k$, $t_1,\dots,t_{\ell_1}$,$t_1',\dots,t_{\ell_2}'$,$t''_1,\dots,t''_{\ell_3}\}$ (not necessarily distinct) with SOS valuations over binary signals.
Let $t_{\ell_1+1}=k$, $t_{\ell_2+1}'=k, t''_{\ell_3+1}=i$,
and $S^* = S'\cup E$. If the following $3+\ell_1+\ell_2+\ell_3$ inequalities hold simultaneously then they all hold with equality:
\begin{itemize}
    \item $v_i(S^*\setminus\{i\}) \geq  v_{t_1}(S^*\setminus\{i\})+v_j(S^*\setminus\{i\})$; 
    
    \item $\forall h\in[\ell_1]: v_{t_h}(S^*\setminus\{i,t_1,...,t_h\}) \geq  v_{t_{h+1}}(S^*\setminus\{i,t_1,...,t_h\}) + v_j(S^*\setminus\{i,t_1,...,t_h\})$; 
    
    \item $v_j(S^*\setminus\{j\}) \geq  v_{t_1'}(S^*\setminus\{j\})+v_i(S^*\setminus\{j\})$;
    
    \item $\forall h\in[\ell_2]: v_{t_h'}(S^*\setminus\{j,t_1',...,t_h'\}) \geq  v_{t_{h+1}'}(S^*\setminus\{j,t_1',...,t_h'\})+v_i(S^*\setminus\{j,t_1',...,t_h'\})$;
    
    \item $v_k(S^*\setminus\{k\}) \geq  v_{t_1''}(S^*\setminus\{k\})+v_j(S^*\setminus\{k\})$;
    
    \item $\forall h\in[\ell_3]: v_{t_h''}(S^*\setminus\{k,t_1'',...,t_h''\}) \geq  v_{t_{h+1}''}(S^*\setminus\{k,t_1'',...,t_h''\})+v_j(S^*\setminus\{k,t_1'',...,t_h''\})$.
\end{itemize}
\end{lemma}

\subsection*{Visualization of Lemma~\ref{lemma2}.} 

Consider the case $\ell_1=\ell_2=\ell_3=0$. 
The first inequality of Lemma~\ref{lemma2} in this case, using that $t_1=t_{\ell_1+1}$ (which equals $k$ by definition), is:
\begin{equation}
    v_i(S^*\setminus\{i\})\geq v_k(S^*\setminus\{i\})+v_j(S^*\setminus\{i\}).\label{eq:block}
\end{equation}
We introduce a visualization of Inequality~\eqref{eq:block} as the $2\times n$ ``block'' shown in Fig.~\ref{fig:block}. The columns of the block correspond to the bidders. 
The first row of the block represents which bidders participate in the inequality (in this case $i,j,k$), with the bidder on the greater (left) side of the inequality depicted in striped red (in this case $i$); the second row represents the signal set in the inequality (in this case $S^*\setminus{i}$), with the signals not in the set depicted in white (in this case $i$). 


\begin{figure}[t]
  \noindent\makebox[\textwidth]{%
  \includegraphics[width=0.5\textwidth]{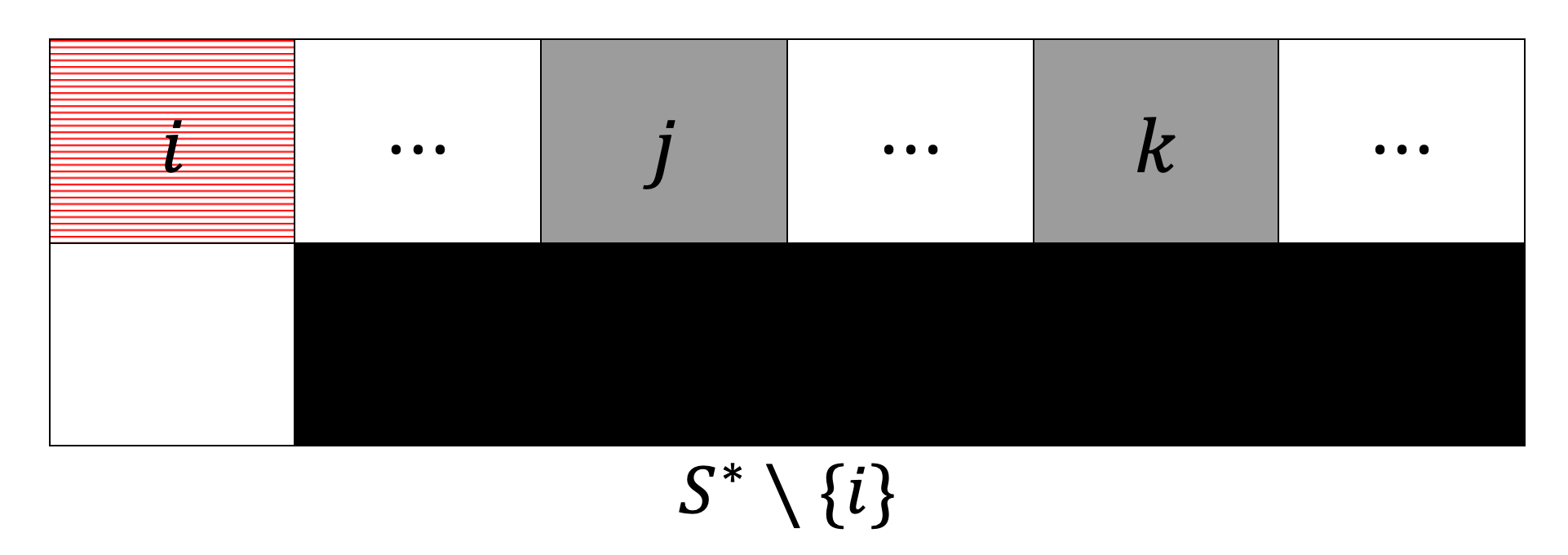}}
  \caption{Visualization of Inequality~\eqref{eq:block}.}
  \label{fig:block}
\end{figure}




\begin{figure}[t]
  \noindent\makebox[\textwidth]{%
  \includegraphics[width=1\textwidth]{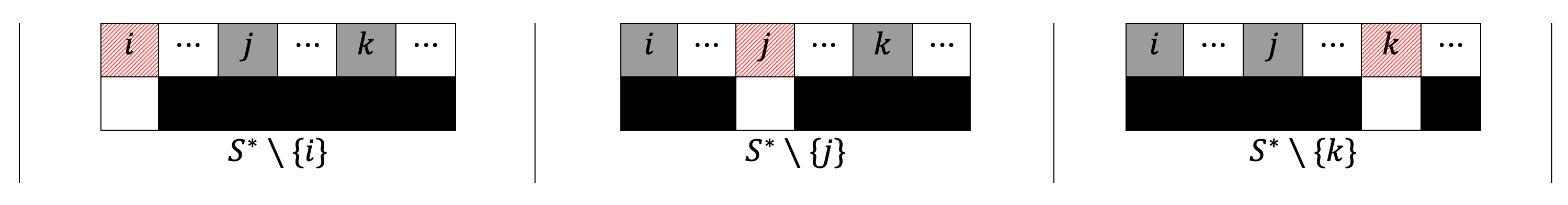}}
  \caption{Visualization of the inequalities of Lemma~\ref{lemma2} for $\ell_1=\ell_2=\ell_3=0$.}
  \label{fig:case0}
\end{figure}


We can use the above visualization
to depict all inequalities of Lemma~\ref{lemma2}. For the case that $\ell_1=\ell_2=\ell_3=0$, these are shown in Fig.~\ref{fig:case0}.
Consider now the case $\ell_1=\ell_2=\ell_3=1$, with the following inequalities (among others): 
\begin{eqnarray*}
v_i(S^*\setminus\{i\})&\geq& v_j(S^*\setminus\{i\})+v_{t_1}(S^*\setminus\{i\});\\
v_{t_1}(S^*\setminus\{i,t_1\})&\geq& v_j(S^*\setminus\{i,t_1\})+v_k(S^*\setminus\{i,t_1\}).
\end{eqnarray*}
In this case we have a fourth bidder $t_1$ who ``bridges'' between $i,j,k$. Instead of an inequality requiring that $v_i\ge v_j+v_k$ directly, here it is required that $v_i\geq v_j+v_{t_1}$, and in turn $v_{t_1} \ge v_j+v_k$ but with a different set of signals. The full system of inequalities for this case (with $6$ inequalities) appears in Fig.~\ref{fig:case1}. 

\begin{figure}[t]
  \noindent\makebox[\textwidth]{%
  \includegraphics[width=1\textwidth]{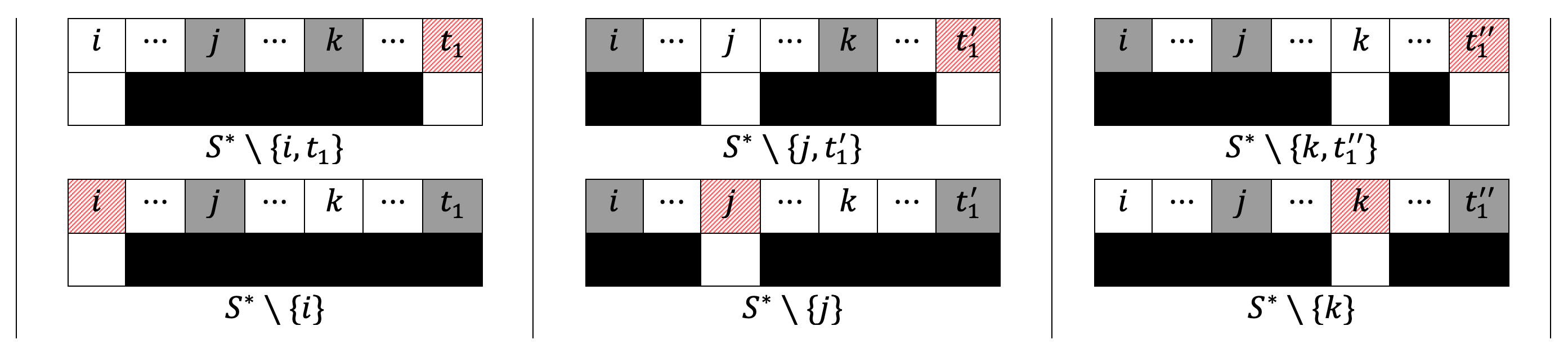}}
  \caption{Visualization of the inequalities of Lemma~\ref{lemma2} for $\ell_1=\ell_2=\ell_3=1$.}
  \label{fig:case1}
\end{figure}

More generally, Lemma~\ref{lemma2} holds for any number of ``bridge'' bidders. The general case is shown in Fig.~\ref{fig:case_full} in Appendix~\ref{appx:lemma2-pf}.



\subsection{When Are Priorities 1 or 2 Guaranteed}
\label{sub:2-approx}

In this section we prove the following lemma (some details of the proof are deferred to Appendix~\ref{appx:2-approx}):

\begin{lemma}
\label{lem:2-approx}
Assume Algorithm~\ref{alg:main} runs on $n$ bidders with SOS valuations over binary signals. Then for every $S\subseteq[n]$, if at the beginning of iteration $S$ there are less than two red bidders, either {\bf Priority~1} or {\bf Priority~2} must hold.  
\end{lemma}

We begin with two observations that will be useful in the proof of Lemma~\ref{lem:2-approx}.

\begin{observation}
\label{obs:no_reds}
If at the beginning of iteration $S$ no bidder is colored {red}, and during the iteration bidder $i$ whose signal is low ($i\notin S$) is colored {red}, then for every pair $j,k\in S$, $v_i(S)\geq v_j(S)+v_k(S)$.
\end{observation}

\begin{proof}
Assume for contradiction that $v_j(S)+v_k(S)>v_i(S)$, then since $j,k$ both have high signals and can be colored {red} at iteration $S$, they have {\bf Priority~1} and should be colored in place of bidder $i$, contradiction.
\qed
\end{proof}

\begin{observation}
\label{obs:one-red}
If at the beginning of iteration $S$ only bidder $t$ whose signal is high ($t\in S$) is colored {red}, and during the iteration bidder $i$ whose signal is low ($i\notin S$) is colored {red}, then for every $j\in S$,  $v_i(S)\geq v_j(S)+v_t(S)$.
\end{observation}

\begin{proof}
Assume for contradiction that $v_j(S)+v_t(S)>v_i(S)$, then since $j,t$ both have high signals and can be colored {red} at iteration $S$ ($t$ is already {red} and $j$ can be colored {red} since there are no other reds besides $t$), they have {\bf Priority~1} and should be colored in place of bidder $i$, contradiction.
\qed
\end{proof}


We can now prove our main lemma; missing details appear in Appendix~\ref{appx:2-approx}.

\begin{proof}[Lemma~\ref{lem:2-approx}, sketch]
%
Fix an iteration $S$ with $<2$ red bidders at its beginning. By \emph{highest bidder} we mean the bidder whose value at $S$ equals $\OPT(S)$.
We split the analysis into cases;
the most challenging cases technically are when the highest bidder is colored black in column $S$, and there are either no {red} cells or a single {red} cell in this column at the beginning of the iteration. Here we focus on the first among these cases and remark at the end how to treat the second, showing in both why a {\bf Priority~1} pair exists in column $S$. The remaining cases are addressed in Appendix~\ref{appx:2-approx}.

\paragraph{Case 1: No red cells.} Assume that at the beginning of iteration $S$, the highest bidder is colored black and there are no {red} cells in column $S$. Denote the highest bidder by $k$ and observe that its color must have propagated backward from $(k,S\cup \{k\})$; let $S^*=S\cup\{k\}$. 
In column~$S^*$ there must therefore be two red bidders, whom we refer to as $i$ and $j$, due to which $k$ is colored black in this column.
Red must have propagated forward to column $S^*$ from $S^*\setminus\{i\}$ and $S^*\setminus\{j\}$. 
Fig.~\ref{fig:prop} shows the allocation status of the relevant bidders at the beginning of iteration~$S$ for subsets $S,S^*,S^*\setminus\{i\},S^*\setminus\{j\}$ -- we use the same visualization as in Section~\ref{sub:main-lemmas}, but with colors in the first row representing those set by the algorithm and arrows representing propagations. 

\begin{figure}[t]
  \noindent\makebox[\textwidth]{%
  \includegraphics[width=0.7\textwidth]{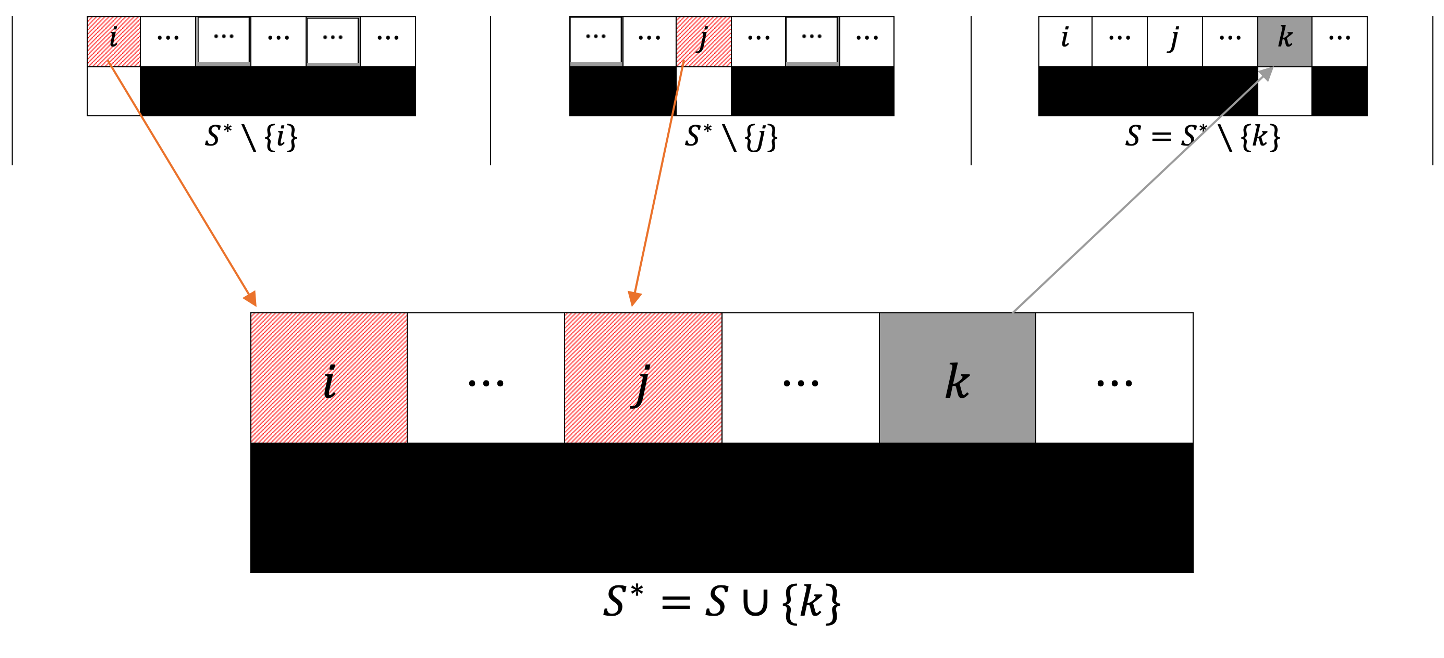}}
  \caption{Colorings at $S, S^*$, $S^*\setminus\{i\}$ and $S^*\setminus\{j\}$ given that $(k,S)$ is black.}
  \label{fig:prop}
\end{figure}

Towards establishing existence of a {\bf Priority~1} pair in column $S$, consider first the case in which the following two conditions hold: 
\begin{enumerate}
        \item At the beginning of iteration $S^*\setminus\{i\}$, no cells in that column are red; 
        \item At the beginning of iteration $S^*\setminus\{j\}$, no cells in that column are red.
\end{enumerate}
By Observation~\ref{obs:no_reds},
\begin{eqnarray}
    v_i(S^*\setminus\{i\})&\geq& v_j(S^*\setminus\{i\})+v_k(S^*\setminus\{i\});\label{eq:cond1}\\
    v_j(S^*\setminus\{j\})&\geq& v_i(S^*\setminus\{j\})+v_k(S^*\setminus\{j\}).\label{eq:cond2}
\end{eqnarray}
If both \eqref{eq:cond1} and \eqref{eq:cond2} hold, 
by Lemma~\ref{lemma2} with $\ell_1=\ell_2=\ell_3=0$
it cannot simultaneously hold that $v_k(S)> v_i(S)+v_j(S)$ (see Fig.~\ref{fig:case0} and related text). Thus $v_i(S)+v_j(S)\geq v_k(S)$, and pair $i,j$ has {\bf Priority~1}.

Now consider the case in which one of the two conditions does not hold, w.l.o.g.~Condition~(1). That is, at the beginning of iteration $S^*\setminus\{i\}$, a bidder $t_1$ is colored {red}. (There is only one such bidder since we know bidder $i$ cannot be red at the beginning of that iteration -- no forward propagation as $i\notin S^*\setminus\{i\}$ -- and that $i$ is colored red during the iteration.) 
By Observation~\ref{obs:one-red}, 
\begin{equation}
    \label{eq:nocond1}
    v_i(S^*\setminus\{i\})\geq v_j(S^*\setminus\{i\})+v_{t_1}(S^*\setminus\{i\}).
\end{equation}
Since $(t_1,S^*\setminus\{i\})$ is {red} at the beginning of iteration $S^*\setminus\{i\}$, the color red necessarily propagated forward from $S^*\setminus\{i,t_1\}$. If Condition~(1) now holds for $S^*\setminus\{i,t_1\}$ then 
by Observation~\ref{obs:no_reds}, 
\begin{equation}
    \label{eq:nocond1_2}
    v_{t_1}(S^*\setminus\{i,t_1\})\geq v_j(S^*\setminus\{i,t_1\})+v_k(S^*\setminus\{i,t_1\})
\end{equation}
If Inequalities~\eqref{eq:cond2}-\eqref{eq:nocond1_2} hold simultaneously,
then by Lemma~\ref{lemma2} with $\ell_1=1,\ell_2=\ell_3=0$ we can again conclude that 
$v_i(S)+v_j(S)\ge v_k(S)$, and pair $i,j$ has {\bf Priority~1}.  
Notice that $t_1$ is the ``bridge'' bidder we discussed in Section~\ref{sub:main-lemmas}. For visualization we note that Inequality~\eqref{eq:cond2} is the one depicted in Fig.~\ref{fig:case0} (middle) while Inequalities \eqref{eq:nocond1}-\eqref{eq:nocond1_2} are shown in Fig.~\ref{fig:case1} (left); these are the inequalities that correspond to $\ell_1=1$ and $\ell_2=0$.

If Condition~(1) does not hold for subset $S^*\setminus\{i,t_1\}$, then there is an additional ``bridge'' bidder $t_2$ that is colored {red} at $S^*\setminus\{i,t_1\}$, and we can possibly apply Lemma~\ref{lemma2} with $\ell_1=2$. If not, we continue in this way until either Condition~(1) holds or only $j,k$ remain in the subset. 
In either case, denote the final number of ``bridge'' bidders by $\ell_1$. 
In the latter case, either Observation~\ref{obs:no_reds} or Observation~\ref{obs:one-red} hold, and so 
\begin{equation*}
    v_{t_{\ell_1}}(\{j,k\}) \geq v_j(\{j,k\})+v_k(\{j,k\}).
\end{equation*}
By applying Lemma~\ref{lemma2} with $\ell_1>0$ (and $\ell_2=\ell_3=0$) we conclude that pair $i,j$ has {\bf Priority~1}.

Observe that the same analysis holds if both Condition~(1) and Condition~(2) are relaxed. In this case Lemma~\ref{lemma2} applies with $\ell_1>0,\ell_2>0$ (and $\ell_3=0$).


\paragraph{Case 2: Single red cell.}
Finally, we address the case in which there exists a red bidder $t_1''$ in column~$S$ at the beginning of iteration $S$. We can write $S$ as $S^*\setminus\{k\}$; the color red of $t_1''$ necessarily propagated forward from $S^*\setminus\{k,t_1''\}$. 
Assume the following third condition holds:
\begin{enumerate}
  \setcounter{enumi}{2}
  \item At the beginning of iteration $S^*\setminus\{k, t_1''\}$, no cells in that column are red.
\end{enumerate}
By Observation~\ref{obs:no_reds}, 
\begin{equation}
    \label{eq:noredsl3}
    v_{t_1''}(S^*\setminus\{k,t_1''\})\geq v_j(S^*\setminus\{k,t_1''\})+v_i(S^*\setminus\{k,t_1''\}).
\end{equation}
If Inequality~\eqref{eq:noredsl3} holds (alongside previous inequalities) then Lemma~\ref{lemma2} applies with $\ell_1>0,\ell_2>0, \ell_3=1$. 
If Condition~(3) does not hold for subset $S^*\setminus\{k,t_1''\}$, 
%
we continue as above, 
denoting the final number of ``bridge'' bidders by $\ell_3$.
%
By applying Lemma~\ref{lemma2} with $\ell_1>0, \ell_2>0$ and $\ell_3>0$, we conclude that pair $i,t_1''$ has {\bf Priority~1}.
\qed
\end{proof}

\section{Summary and Future Directions}
\label{sec:summary}

Tension between optimization and truthfulness is an important theme of algorithmic game theory. With interdependent values, this tension appears even without computational considerations. Since with interdependence there is an inherent clash between welfare maximization and truthfulness, the approximation toolbox comes in handy. We apply it to arguably the simplest possible setting (single-item auctions with binary signals), and get a tight understanding of the tradeoff (i.e., what fraction of the optimal welfare can be guaranteed by a truthful mechanism). Our results extend beyond single items.

Two promising future directions are: (i)~generalizing our results beyond binary signals, and (ii)~designing an ``on the fly'' tractable version of the 2-approximation truthful mechanism (i.e., a version that gets signal reports and returns an allocation only for the reported signal profile). For the former direction, 
non-binary signals pose additional challenges since two priorities are no longer sufficient in the algorithm, and additionally the propagation is more complex. In Appendix~\ref{over-binary} we present progress towards resolving these challenges (in particular, Lemma~\ref{lemma-int-signals} extends Lemma~\ref{lemma2} to functions over general integer signals).





%
%
%
%
\bibliographystyle{splncs04}
\bibliography{main.bib}

\newpage
\appendix

\section{Main Algorithm}
\label{appx:pseudocode}

In this section we give the pseudocode of our main algorithm (Algorithm~\ref{alg:main}), which proves Theorem~\ref{thm:main} by returning a truthful, $2$-approximate welfare maximizing mechanism for any given setting $v_1,\dots,v_n$. The algorithm's main procedure is \textsc{Allocate}, and it contains two sub-procedures \textsc{ColorRed} and \textsc{ColorBlack}. For completeness, in Appendix~\ref{appx:complexity} we analyze the algorithm's running time as a function of $n$.  

\begin{algorithm}                     
\begin{algorithmic}[1]
\caption{\label{alg:main} Mechanism Construction}

\State \textit{\textbf{Input:} Set of bidders $E=\{1,\dots,n\}$ and their valuation functions $V=(v_1,\dots,v_n)$, where $value(b,S)$ for $b\in E$ and $S\subseteq[n]$ returns $v_b(S)$.}
\State \textit{\textbf{Output:} Allocation rule $x$, where allocation$(b,S)$ returns $x_b(S)$ represented by the following colors -- red represents $x_b(S)=\frac{1}{2}$ and black represents $x_b(S)=0$.}
\State \textit{\textbf{Initialization:} For every $b,S$, initialize allocation$(b,S)$=white.
Let $S_1,S_2,\dots,S_{2^n}$ be an inclusion-compatible ordering of all subsets of $[n]$.
}
\\
\Function{Allocate}{$E,V$}
\For{$S\in (S_1,S_2,\dots,S_{2^n})$}
\State{cnt\_{reds} $\gets \left|\{b\in E\text{ s.t. allocation}(b, S)= \text{red}\}\right|$}
\If{cnt\_{reds} $= 2$}
\State{\textit{continue;}}\Comment{go to next subset}
\EndIf
\State max\_v $\gets$ max\{$value(b, S)$ : $b\in E$\}
\Comment{max\_v  $=\OPT(S)$}
\If{$\exists b_1,b_2$ s.t.~$b_1\neq b_2$ and 
$\{b_1,b_2\}\subseteq S$}
\If{allocation($b_1$, S) $\neq$ black and allocation($b_2$, S) $\neq$ black} 
\If{$\nexists b_3\notin \{b_1,b_2\}$ 
s.t. allocation($b_3$, S)$=$red}
\If{$value(b_1,S)+value(b_2,S)\geq$ max\_v}
\State{\Call{ColorRed}{$b_1$, $S$}} \Comment{\textbf{Priority 1}}
\State{\Call{ColorRed}{$b_2$, $S$}} \Comment{\textbf{Priority 1}}
\State{\textit{continue;}}
\EndIf
\EndIf
\EndIf
\EndIf

\If{$\exists b_1$ s.t. allocation($b_1$, $S$) $\neq$ black}
\If{$value(b_1, S)=$ max\_v}
\State{\Call{ColorRed}{$b_1$, $S$}} \Comment{\textbf{Priority 2}}
\For{$b_2$ s.t. allocation$(b_2, S)=$ white}
\State{\Call{ColorBlack}{$b_2$, $S$}} 
\EndFor
\EndIf
\EndIf
\EndFor
\EndFunction
\algstore{myalg}
\end{algorithmic}
\end{algorithm}

\begin{algorithm}
\small
\begin{algorithmic}[1]
\algrestore{myalg}
\Function{ColorRed}{{$b$}, {$S$}}
\If{allocation($b$, $S$) $=$ black}
\State Error(``Cannot color a black cell red")
\EndIf
\If{allocation($b$, $S$) $=$ red}
\State \Return
\EndIf
\If{$\exists b_1,b_2$ s.t.~$b_1,b_2\neq b$ and $b_1\neq b_2$, and allocation($b_1$, $S$) $=$ allocation($b_2$, $S$) $=$ red}
\State Error(``Cannot color more than two cells red")
\EndIf
\State allocation($b$, $S$) $\gets$ red
\If{$\exists b_1\neq b$ s.t. allocation($b_1$, $S$) $=$ red}
\Comment{\textbf{Two cells $b,b_1$ are red}}
\For{$b'\neq b,b_1$}
\State \Call{ColorBlack}{$b'$, $S$}
\Comment{\textbf{Color the others black}}
\EndFor
\EndIf
\If{$b\notin S$}
\State \Call{ColorRed}{$b$, $S\cup\{b\}$}
\Comment{\textbf{Propagate forward}}
\EndIf
\EndFunction
\Comment{\newline -----------------------------------------------------------------------------------------------------------}
\Function{ColorBlack}{{$b$}, {$S$}}
\If{allocation($b$, $S$) $=$ red}
\State Error(``Cannot color a red cell black")
\EndIf
\If{allocation($b$, $S$) $=$ black}
\State \Return
\EndIf
\State allocation($b$, $S$) $\gets$ black
\If{$b\in S$}
\State \Call{ColorBlack}{$b$, $S\setminus\{b\}$}
\Comment{\textbf{Propagate backward}}
\EndIf
\EndFunction

\end{algorithmic}
\end{algorithm}

\subsection{Computational Complexity}
\label{appx:complexity}

The loop in line 36 iterates over all subsets of bidders, that is, over $2^n$ sets. For each set:
\begin{enumerate}
    \item We count the {red} bidders in line $7$ which takes $O(n)$;
    \item We find the maximum value in line $11$ which takes $O(n)$;
    \item In lines $12$-$15$ we search for two bidders that satisfy {\bf Priority 1}, for which we iterate over all bidders once more to check there is no other {red} bidder. This takes $O(n^3)$;
    \item When two bidders are found, we call \textsc{ColorRed} twice. In \textsc{ColorRed} we check if two are already {red}, which takes $O(n)$, and if so we color all the others {black} with \textsc{ColorBlack} (which takes O(1)). Thus, \textsc{ColorRed} in total takes $O(n)$;
    \item In lines $23$-$26$ we iterate  over all bidders and search for a bidder with {\bf Priority 2}. If found we call \textsc{ColorRed}, which takes $O(n)$;
\end{enumerate}
In total, the algorithm's runtime complexity is $O(n^3\cdot 2^n)$.

\section{Missing Proofs from Section~\ref{sec:analysis}}
\label{appx:missing}


\subsection{Proof of Lemma~\ref{lemma2}}
\label{appx:lemma2-pf}
\label{appx:2-approx}
\begin{figure}[t]
  \noindent\makebox[\textwidth]{%
  \includegraphics[width=1\textwidth]{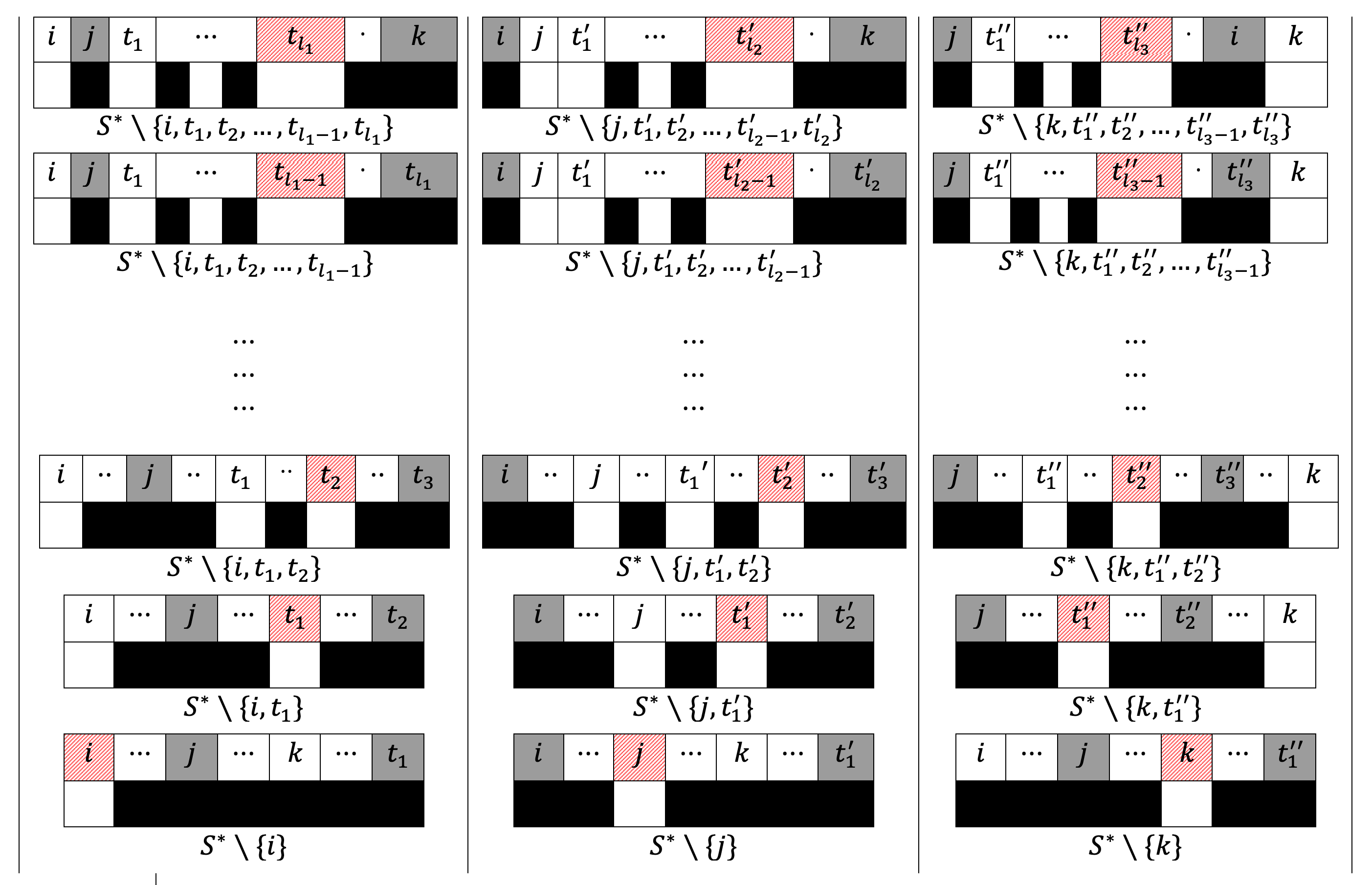}}
  \caption{Visualization of the inequalities of Lemma~\ref{lemma2} for the general case.}
  \label{fig:case_full}
\end{figure}
In this section we prove that if the inequalities of Lemma~\ref{lemma2} (as depicted in Fig.~\ref{fig:case_full}) hold simultaneously, then they must all hold with equality.

\begin{proof}
The proof is by induction. 
For the base cases define $u_{i'}(x)=v_{i'}(x\mid S')$ for every ${i'}\in[n]$. Note that we provide the analysis of Base Cases~2-4 to build intuition (Base Case 1 is sufficient to complete the proof).

\subsubsection*{Base cases.} 

\paragraph{Base Case 1: $\ell_1=\ell_2=\ell_3=0$.}  For any three bidders $i,j,k$, consider the following three inequalities:
\begin{itemize}
    \item $u_i(jk)+v_i(S')\geq u_j(jk)+v_j(S')+u_k(jk)+v_k(S')$;
    \item $u_j(ik)+v_j(S')\geq u_i(ik)+v_i(S')+u_k(ik)+v_k(S')$;
    \item $u_k(ij)+v_k(S')\geq u_i(ij)+v_i(S')+u_j(ij)+v_j(S')$.
\end{itemize}

We sum the  inequalities and replace $v_i(S')+v_j(S')+v_k(S')$ by $A$ to get: 
%
%
$$
u_i(jk)+u_j(ik)+u_k(ij)+A\geq u_j(jk)+u_k(jk)+u_i(ik)+u_k(ik)+u_i(ij)+u_j(ij)+2A.
$$
Since all factors are non-negative, $A\leq 2A$ and so: $$
u_i(jk)+u_j(ik)+u_k(ij)\geq u_j(jk)+u_k(jk)+u_i(ik)+u_k(ik)+u_i(ij)+u_j(ij).
$$
By subadditivity we know that:
\begin{itemize}
    \item $u_i(j)+u_i(k)\geq u_i(jk)$;
    \item $u_j(i)+u_j(k)\geq u_j(ik)$;
    \item $u_k(i)+u_k(j)\geq u_k(ij)$.
\end{itemize}
Thus, we get: $$\underbrace{u_i(j)}_\text{1}+\underbrace{u_i(k)}_\text{2}+\underbrace{u_j(i)}_\text{3}+\underbrace{u_j(k)}_\text{4}+\underbrace{u_k(i)}_\text{5}+\underbrace{u_k(j)}_\text{6}\geq$$$$ \underbrace{u_j(jk)}_\text{4}+\underbrace{u_k(jk)}_\text{6}+\underbrace{u_i(ik)}_\text{2}+\underbrace{u_k(ik)}_\text{5}+\underbrace{u_i(ij)}_\text{1}+\underbrace{u_j(ij)}_\text{3}$$
Notice that for each pair of summands with the same number, the left-hand side is at most the right-hand side by monotonicity. 
This completes the analysis of Base Case 1.

\paragraph{Base Case 2: $\ell_1=1, \ell_2=\ell_3=0$.} For any four bidders $i,j,k,t$, consider the following four inequalities:
\begin{itemize}
    \item $u_i(jkt)+v_i(S')\geq u_j(jkt)+v_j(S')+u_t(jkt)+v_t(S')$;
    \item $u_t(jk)+v_t(S')\geq u_j(jk)+v_j(S')+u_k(jk)+v_k(S')$;
    \item $u_j(ikt)+v_j(S')\geq u_i(ikt)+v_i(S')+u_k(ikt)+v_k(S')$;
    \item $u_k(ijt)+v_k(S')\geq u_i(ijt)+v_i(S')+u_j(ijt)+v_j(S')$.
\end{itemize}

We sum the inequalities 
and use that $v_i(S')+v_t(S')+v_j(S')+v_k(S')\leq 3v_j(S')+v_t(S')+2v_k(S')+2v_i(S')$ to get:
\begin{eqnarray*}
&u_i(jkt)+u_t(jk)+u_j(ikt)+u_k(ijt) \geq&\\ &u_j(jkt)+u_t(jkt)+u_j(jk)+u_k(jk)+u_i(ikt)+u_k(ikt)+u_i(ijt)+u_j(ijt).&
\end{eqnarray*}
By subadditivity we know that:
\begin{itemize}
    \item $u_i(j)+u_i(kt)\geq u_i(jkt)$;
    \item $u_j(i)+u_j(k)+u_j(t)\geq u_j(ikt)$;
    \item $u_k(it)+u_k(j)\geq u_k(ijt)$.
\end{itemize}
Thus, we get:
$$\underbrace{u_i(j)}_\text{1}+\underbrace{u_i(kt)}_\text{2}+\underbrace{u_t(jk)}_\text{3}+\underbrace{u_j(i)}_\text{4}+\underbrace{u_j(k)}_\text{5}+\underbrace{u_j(t)}_\text{6}+\underbrace{u_k(it)}_\text{7}+\underbrace{u_k(j)}_\text{8}\geq$$$$ \underbrace{u_j(jkt)}_\text{6}+\underbrace{u_t(jkt)}_\text{3}+\underbrace{u_j(jk)}_\text{5}+\underbrace{u_k(jk)}_\text{8}+\underbrace{u_i(ikt)}_\text{2}+\underbrace{u_k(ikt)}_\text{7}+\underbrace{u_i(ijt)}_\text{1}+\underbrace{u_j(ijt)}_\text{4}$$
As above, for each pair of summands with the same number, the left-hand side is at most the right-hand side by monotonicity. This completes the analysis of Base Case 2.

\paragraph{Base Case 3: $\ell_2=1, \ell_1=\ell_3=0$.}
The analysis of this case is symmetric to that of $\ell_1=1, \ell_2=\ell_3=0$.

\paragraph{Base Case 4: $\ell_1=\ell_2=0, \ell_3=1$.}
For any four bidders $i,j,k,t''_1$, consider the following four inequalities:
\begin{enumerate}
    \item $u_i(jkt''_1)+v_i(S')\geq u_j(jkt''_1)+v_j(S')+u_k(jkt''_1)+v_k(S')$;
    \item $u_j(ikt''_1)+v_j(S')\geq u_i(ikt''_1)+v_i(S')+u_k(ikt''_1)+v_k(S')$;
    \item $u_{t''_1}(ij)+v_{t''_1}(S')\geq u_j(ij)+v_j(S')+u_i(ij)+v_i(S')$;
    \item $u_k(ijt''_1)+v_k(S')\geq u_{t''_1}(ijt''_1)+v_{t''_1}(S')+u_j(ijt''_1)+v_j(S')$.
\end{enumerate}

We sum the inequalities and use that
$v_i(S')+v_j(S')+v_{t''_1}(S')+v_k(S')\leq 2v_i(S')+3v_j(S')+v_{t''_1}(S')+2v_k(S')$ to get:
\begin{eqnarray*}
&u_i(jkt''_1)+u_j(ikt''_1)+u_{t''_1}(ij)+u_k(ijt''_1) \geq&\\
&u_j(jkt''_1)+u_k(jkt''_1)+u_i(ikt''_1)+u_k(ikt''_1)+u_j(ij)+u_i(ij)+u_{t''_1}(ijt''_1)+u_j(ijt''_1).&
\end{eqnarray*}
By subadditivity we know that:
\begin{itemize}
    \item $u_i(j)+u_i(k{t''_1})\geq u_i(jk{t''_1})$;
    \item $u_j(i)+u_j(k)+u_j({t''_1})\geq u_j(ik{t''_1})$;
    \item $u_k(i{t''_1})+u_k(j)\geq u_k(ij{t''_1})$.
\end{itemize}
Thus, we get:
$$
\underbrace{u_i(j)}_\text{1}+\underbrace{u_i(k{t''_1})}_\text{2}+\underbrace{u_{t''_1}(ij)}_\text{3}+\underbrace{u_j(i)}_\text{4}+\underbrace{u_j(k)}_\text{5}+\underbrace{u_j(t''_1)}_\text{6}+\underbrace{u_k(i{t''_1})}_\text{7}+\underbrace{u_k(j)}_\text{8}\geq$$$$ \underbrace{u_j(ij{t''_1})}_\text{6}+\underbrace{u_{t''_1}(ijt''_1)}_\text{3}+\underbrace{u_j(jkt''_1)}_\text{5}+\underbrace{u_k(jkt''_1)}_\text{8}+\underbrace{u_i(ikt''_1)}_\text{2}+\underbrace{u_k(ikt''_1)}_\text{7}+\underbrace{u_i(ij)}_\text{1}+\underbrace{u_j(ij)}_\text{4}.
$$
As above, for each pair of summands with the same number, the left-hand side is at most the right-hand side by monotonicity. This completes the analysis of Base Case 4.

\subsubsection*{Induction step.} 

Assuming the lemma holds for $\ell_1-1,\ell_2,\ell_3$, we prove it for $\ell_1, \ell_2, \ell_3$. The proof is symmetric for increasing $\ell_2$ or $\ell_3$.

Define $S=S'\cup\{t_{\ell_1}\}$ and $u_{i'}^*(x)=v_{i'}(x \mid S)$ for every ${i'}\in[n]$. Using this notation we can rewrite the inequalities of the lemma as follows:
\begin{itemize}
    \item $u_i^*(jkt_1t_2...t_{\ell_1-1}t_1't_2'...t_{\ell_2}'t''_1,...,t''_{\ell_3}) + v_i(S) \geq\\ u_j^*(jkt_1t_2...t_{\ell_1-1}t_1't_2'...t_{\ell_2}'t''_1,...,t''_{\ell_3}) + v_j(S) +\\ u_{t_1}^*(jkt_1t_2...t_{\ell_1-1}t_1't_2'...t_{\ell_2}'t''_1,...,t''_{\ell_3}) + v_{t_1}(S)$;
    
    \item $\forall h\in[\ell_1] : u_{t_h}^*(jkt_{h+1}...t_{\ell_1-1}t_1't_2'...t_{\ell_2}'t''_1,...,t''_{\ell_3}) + v_{t_h}(S) \geq\\ u_j^*(jkt_{h+1}...t_{\ell_1-1}t_1't_2'...t_{\ell_2}'t''_1,...,t''_{\ell_3}) + v_j(S) +\\ u_{t_{h+1}}^*(jkt_{h+1}...t_{\ell_1-1}t_1't_2'...t_{\ell_2}'t''_1,...,t''_{\ell_3}) + v_{t_{h+1}}(S)$;
    
    \item $u_j^*(ikt_1t_2...t_{\ell_1-1}t_1't_2'...t_{\ell_2}'t''_1,...,t''_{\ell_3}) + v_j(S) \geq\\ u_i^*(ikt_1t_2...t_{\ell_1-1}t_1't_2'...t_{\ell_2}'t''_1,...,t''_{\ell_3}) + v_i(S)+\\ u_{t_1'}^*(jkt_1t_2...t_{\ell_1-1}t_1't_2'...t_{\ell_2}'t''_1,...,t''_{\ell_3})+v_{t_1'}(S)$;
    
    \item $\forall h\in[\ell_2] : u_{t_h'}^*(ikt_1...t_{\ell_1-1}t_{h+1}'...t_{\ell_2}'t''_1,...,t''_{\ell_3}) + v_{t_h'}(S) \geq\\ 
    u_i^*(ikt_1...t_{\ell_1-1}t_{h+1}'...t_{\ell_2}'t''_1,...,t''_{\ell_3}) + v_i(S) +\\ u_{t_{h+1}'}^*(jkt_1...t_{\ell_1-1}t_{h+1}'...t_{\ell_2}'t''_1,...,t''_{\ell_3}) + v_{t_{h+1}'}(S)$;
    
    \item $u_k^*(ijt_1t_2...t_{\ell_1-1}t_1't_2'...t_{\ell_2}'t''_1,...,t''_{\ell_3}) + v_k(S) \geq\\ 
    u_i^*(ijt_1t_2...t_{\ell_1-1}t_1't_2'...t_{\ell_2}'t''_1,...,t''_{\ell_3}) + v_i(S) +\\ u_j^*(ijt_1t_2...t_{\ell_1-1}t_1't_2'...t_{\ell_2}'t''_1,...,t''_{\ell_3}) + v_j(S)$;
    
    \item $\forall h\in[\ell_3] : u_{t''_h}^*(ijt_1...t_{\ell_1-1}t_1't_2'...t_{\ell_2}'t''_{h+1},...,t''_{\ell_3}) + v_{t''_h}(S)\geq\\ u_j^*(ijt_1...t_{\ell_1-1}t_1't_2'...t_{\ell_2}'t''_{h+1},...,t''_{\ell_3}) + v_j(S) +\\ u_{t''_{h+1}}^*(ijt_1...t_{\ell_1-1}t_1't_2'...t_{\ell_2}'t''_{h+1},...,t''_{\ell_3})+v_{t_{h+1}}(S)$.
\end{itemize}
The last two inequalities in the second bullet are:
\begin{eqnarray}
    &u_{t_{\ell_1-1}}^*(jkt_1't_2'...t_{\ell_2}') + v_{t_{\ell_1-1}}(S) \geq&\nonumber\\ &u_j^*(jkt_1't_2'...t_{\ell_2}') + v_j(S) + u_{t_{\ell_1}}^*(jkt_1't_2'...t_{\ell_2}') + v_{t_{\ell_1}}(S);&\label{eq:before-last}\\
    &u_{t_{\ell_1}}^*(jkt_1't_2'...t_{\ell_2}') + v_{t_{\ell_1}}(S)\geq&\nonumber\\ &u_j^*(jkt_1't_2'...t_{\ell_2}') + v_j(S) + u_{k}^*(jkt_1't_2'...t_{\ell_2}') + v_{k}(S).&\label{eq:last}
\end{eqnarray}
Combining Inequalities~\eqref{eq:before-last}-\eqref{eq:last} we get that:
\begin{eqnarray}
&u_{t_{\ell_1-1}}^*(jkt_1't_2'...t_{\ell_2}') +v_{t_{\ell_1-1}}(S)\geq&\nonumber\\
&u_j^*(jkt_1't_2'...t_{\ell_2}')+v_j(S)+u_j^*(jkt_1't_2'...t_{\ell_2}')+v_j(S)+u_{k}^*(jkt_1't_2'...t_{\ell_2}')+v_{k}(S)\geq&\nonumber\\
&u_j^*(jkt_1't_2'...t_{\ell_2}')+v_j(S)+u_{k}^*(jkt_1't_2'...t_{\ell_2}')+v_{k}(S).&\label{eq:unified}
\end{eqnarray}
%
Consider replacing the last two inequalities in the second bullet with the combined Inequality~\eqref{eq:unified}. We know from the induction assumption for $\ell_1-1, \ell_2, \ell_3$ that these $\ell_1+\ell_2+\ell_3+2$ inequalities must hold with equality. Thus, the same holds for the $\ell_1+\ell_2+\ell_3+3$ original inequalities.
\qed
\end{proof}

\subsection{ Proof of Lemma~\ref{lem:2-approx}: Supplementary}

\begin{proof}[Lemma~\ref{lem:2-approx}, missing details]
Fix a subset $S\subseteq[n]$ such that there are $<2$ red bidders at the beginning of iteration $S$. We show by case analysis that one of the two priorities holds. 

Denote the highest bidder (whose value at $S$ equals $\OPT(S)$) by $k$. If at the beginning of iteration $S$ bidder $k$ is not colored black, then $k$ has {\bf Priority~2} (note there is also possibly an additional pair with {\bf Priority~1}), and this is sufficient to complete the proof. 
%
%
The main technical challenge is if at the beginning of iteration $S$, bidder $k$ \emph{is} colored black. Observe this can only happen due to back propagation from $(k,S\cup \{k\})$. 

Let $S^* = S\cup \{k\}$.
Since $S\subset S^*$ we know that at the beginning of iteration~$S$, Algorithm~\ref{alg:main} did not yet reach iteration $S^*$. 
Thus, for back propagation to occur, two reds must have propagated forward to $S^*$;
denote the corresponding red bidders by $i,j$. The forward propagation must have been from $(i,S^*\setminus\{i\})$ and $(j,S^*\setminus\{j\})$, resulting in coloring $(k,S^*)$ black (and back-propagating the color black to $(k,S)$). 
We now split the analysis into two cases, by the number of cells colored red at the beginning of iteration $S$ (which is $<2$ by assumption). We use $i,j,k$ as defined above in the case analysis: 
%
\paragraph{Case 1: No red cells.}
As detailed in Section~\ref{sub:2-approx},
Observations~\ref{obs:no_reds}-\ref{obs:one-red} imply that for $\ell_1,\ell_2\geq 0$ there exist $t_1,t_2,...,t_{l_1},t_1',t_2',...,t_{l_2}'$ for which the following inequalities simultaneously hold:
\begin{itemize}
    \item $v_i(S^*\setminus\{i\}) \geq  v_{t_1}(S^*\setminus\{i\})+v_j(S^*\setminus\{i\})$; 
    
    \item $\forall h\in[\ell_1]: v_{t_h}(S^*\setminus\{i,t_1,...,t_h\}) \geq  v_{t_{h+1}}(S^*\setminus\{i,t_1,...,t_h\}) + v_j(S^*\setminus\{i,t_1,...,t_h\})$; 
    
    \item $v_j(S^*\setminus\{j\}) \geq  v_{t_1'}(S^*\setminus\{j\})+v_i(S^*\setminus\{j\})$;
    
    \item $\forall h\in[\ell_2]: v_{t_h'}(S^*\setminus\{j,t_1',...,t_h'\}) \geq  v_{t_{h+1}'}(S^*\setminus\{j,t_1',...,t_h'\})+v_i(S^*\setminus\{j,t_1',...,t_h'\})$;
    
\end{itemize}
By Lemma~\ref{lemma2} with $\ell_3=0$, 
if $v_k(S^*\setminus\{k\}) \geq  v_{i}(S^*\setminus\{k\})+v_j(S^*\setminus\{k\})$
then the inequalities all hold with equality. In particular, $v_k(S)=v_k(S^*\setminus\{k\}) =  v_{i}(S^*\setminus\{k\})+v_j(S^*\setminus\{k\})=v_i(S)+v_j(S)$.
In other words, by Lemma~\ref{lemma2}
either $v_k(S)=v_i(S)+v_j(S)$, or else $v_k(S)< v_i(S)+v_j(S)$ to begin with.
We conclude that
$v_i(S)+v_j(S)\ge v_k(S)$.
Therefore, both $i$ and $j$ are contained in the set $S$, their aggregate value is at least that of $k$, and there are no bidders already colored red at $S$. Thus the pair $i,j$ has {\bf Priority 1}. 

\paragraph{Case 2: Single red cell.} Denote the bidder colored red at $S$ by $t''_1$. 
As explained in Section~\ref{sub:2-approx},
Observations~\ref{obs:no_reds}-\ref{obs:one-red} imply that for $\ell_3\geq 0$ there exist 
$t''_2,...,t''_{\ell_3+1}$ where $t''_{\ell_3+1}=i$ such that:
%
$$\forall h\in [\ell_3]: v_{t''_h}(S\setminus\{t''_1,..,t''_h\})\geq v_{t''_{h+1}}(S\setminus\{t''_1,..,t''_h\})+v_j(S\setminus\{t''_1,..,t''_h\}).$$
These inequalities hold simultaneously with the inequalities for $\ell_1,\ell_2$ above.
By Lemma~\ref{lemma2}, 
if $v_k(S^*\setminus\{k\}) \geq  v_{t''_1}(S^*\setminus\{k\})+v_j(S^*\setminus\{k\})$
then the inequalities all hold with equality. As above we conclude that 
$v_{t''_1}(S)+v_j(S)\ge v_k(S)$, and that the pair $t''_1,j$ has {\bf Priority 1}.
%
\qed
\end{proof}

\section{Matroid Auction Settings}
\label{matroid}

In this appendix we establish, by reduction to the single item case, the existence of a truthful mechanism for matroid settings that achieves a $2$-approximation to the optimal welfare. This existence result follows directly from Proposition~\ref{pro:matroid} below.

A \emph{matroid auction setting} is defined by a matroid $([n],\mathcal{I})$, where $\mathcal{I}$ contains all \emph{independent sets} of bidders, i.e., subsets of bidders who can simultaneously \emph{win} in the auction. 
We also refer to sets in $\mathcal{I}$ as \emph{feasible}. For example, if there are $k$ units of the item to allocate in the auction, $\mathcal{I}$ can be all possible subsets of $k$ bidders. The rest of the setting is as before, i.e., every bidder $i$ has a signal $s_i$ and a valuation $v_i(s_1,\dots,s_n)$ for winning. It is well known that in matroid settings \emph{with given values}, welfare maximization is achieved by greedily adding bidders to the winner set $W$ while keeping $W$ feasible. We refer to this algorithm as Greedy.

\begin{algorithm}         
\begin{algorithmic} [1]
\caption{Matroid Settings}
\Function{MatroidMechanism}{$(E,\mathcal{I}),k,S=\{s_i\mid s_i=1\},V=(v_1,\dots,v_n)$}
\State{$W=\emptyset$}\Comment{Current winning bidders}
\For{$j\in [k]$}\Comment{Run $k$ times where $k$ is the matroid rank}
\State{$x=$ \Call{Allocate}{$E$,$V$}}
\State{$b=$ randomly choose a bidder such that every bidder $b'\in E$ is chosen with probability $x(b',S)$}
\State{$W=W\cup\{b\}$}
\Comment{Add $b$ to the winning set}
\State{$E=$ all bidders $b'$ such that $W\cup\{b'\}\in \mathcal{I}$}
\Comment{Remove bidders who can no longer feasibly win}
\EndFor
\EndFunction
\label{alg:matroid}
\end{algorithmic}
\end{algorithm}


\begin{proposition}
\label{pro:matroid}
If for single item auctions there exists a monotone feasible allocation rule that achieves a 2-approximation to the optimal welfare, then there exists such an allocation rule for matroid auction settings as well.
\end{proposition}

\begin{proof}
Consider the matroid auction setting and fix a signal profile $S$. Let $k$ be the rank of the matroid. 
Algorithm~\ref{alg:matroid} gives our reduction.
For the analysis we use the following notation:
\begin{itemize}
    \item $\GRD_i$: The welfare contribution (i.e., the valuation given $S$) of the $i$th bidder chosen by Greedy. Using this notation, $\OPT=\sum_{i=1}^k \GRD_i$. 
    \item $\ALG_i$: The expected welfare contribution of the $i$th bidder chosen by Algorithm~\ref{alg:matroid} (where the expectation is over the randomness of \textsc{Allocate}). Using this notation and linearity of expectation, $\ALG=\sum_{i=1}^k \ALG_i$. 
    \item $\GRD_i^*$: After choosing $i-1$ winners by Algorithm~\ref{alg:matroid}, $\GRD_i^*$ is the welfare contribution of the $i$th bidder chosen by Greedy.
\end{itemize}
To show that $\ALG\geq \frac{1}{2}\OPT$, it is sufficient to show that 
for every $i\in[k]$, $\ALG_i\geq \frac{1}{2}\GRD_i^*\geq \frac{1}{2}\GRD_i$:
\begin{enumerate}
    \item $\ALG_i\geq \frac{1}{2}\GRD_i^*$: Let $v^*$ be the highest value among the remaining bidders after $i-1$ iterations of Algorithm~\ref{alg:matroid}. Observe that $\GRD_i^*= v^*$. Since \textsc{Alloc} guarantees a 2-approximation, $\ALG_i\geq \frac{1}{2}v^*=\frac{1}{2}\GRD_i^*$. 
    \item $\GRD_i^*\ge \GRD_i$: Let $W_{i-1}$ be the winning set of bidders after $i-1$ iterations of Algorithm~\ref{alg:matroid}, and let $W^*_{i}$ be the winning set of bidders after $i$ iterations of Greedy. Observe that $\GRD_i$ is the lowest among the values of bidders in $W^*_i$. 
    By the exchange property of matroids, there is a bidder in $W^*_i$ that can be feasibly added to $W_{i-1}$, and $\GRD_i^*$ will be the highest value of such a bidder. The inequality follows. 
\end{enumerate}

{It remains to show that the algorithm maintains monotonicity. Denote: \begin{itemize}
    \item $W_i$: Indicator -- bidder $i$ wins at one of the iterations.
    \item $W_{i}^{t}$: Indicator -- bidder $i$ wins at iteration $t$.
\end{itemize}
We need to show that $\mathop{\mathbb{E}}{[W_{i} \mid s_i=1]}\geq \mathop{\mathbb{E}}{[W_i \mid s_i=0]}$.
Note that:
\begin{eqnarray}
\mathop{\mathbb{E}}{[W_{i} \mid s_i=1]}=&\mathop{\mathbb{E}}[\sum_{t=1}^{k}{W_i^t}\mid s_i=1]& \label{eq:expec_sum}\\ 
=&\sum_{t=1}^{k}{\mathop{\mathbb{E}}{[W_i^t \mid s_i=1}]}& \label{eq:linearity_of_expec_one} \\ \geq
&\sum_{t=1}^{k}{\mathop{\mathbb{E}}{[W_i^t \mid s_i=0}]}& \label{eq:matroid_to_single_monotone}\\=
&\mathop{\mathbb{E}}{[\sum_{t=1}^{k}{{W_i^t} \mid s_i=0]}]}& \label{eq:linearity_of_expec_zero}\\
=&\mathop{\mathbb{E}}{[W_{i} \mid s_i=0]},& \label{eq:expec_sum_zero}
\end{eqnarray}
where Equations~\eqref{eq:expec_sum} and~\eqref{eq:expec_sum_zero} are due to the fact that the bidder can only win at one of the $k$ iterations, equations~\eqref{eq:linearity_of_expec_one} and~\eqref{eq:linearity_of_expec_zero} are due to the linearity of expectation and inequality~\eqref{eq:matroid_to_single_monotone} is due to the fact that \textsc{Allocate} at each iteration is monotone; i.e. by Theorem~\ref{thm:main}}.
This completes the proof. \qed
\end{proof}


\section{Beyond Binary Signals}
\label{over-binary}

An interesting open question is whether our result extends beyond binary signals. This appendix gives two indications that the answer may be positive.

\subsection{Simulations}
We ran a computer simulation that constructs random values, which satisfy SOS for a setting with 3 bidders and 3 signals each. The values are randomized under the SOS and monotonicity (of values) constraints. We achieve that by constructing the values iteratively by the natural signals order to hold under the constraints. To illustrate this, consider the valuation $v_1(s=(1,1,0))$. 
\begin{enumerate}
    \item We start by constructing the valuation $v_1(0,0,0)$ by simply choosing a random number;
    \item We construct the valuations $v_1(1,0,0)$ and $v_1(0,1,0)$ to random numbers in the range $[v_1(0,0,0), \infty]$ to satisfy the constraint of monotonicity of the valuations; i.e. $v_1(0,1,0)\geq v_1(0,0,0)$;
    \item We finally construct the valuation $v_1(1,1,0)$ to a random number in the range $[\max\{v_1(1,0,0), v_1(0,1,0)\}, v_1(1,0,0) + v_1(0,1,0) - v_1(0,0,0)]$;
\end{enumerate}
The upper bound of the range in (3) is used to satisfy the SOS constraint. By SOS we know that: $$v_1(1,1,0) - v_1(1,0,0)\leq v_1(0,1,0) - v_1(0,0,0)$$
Reordering the inequality gives the upper bound: $$v_1(1,1,0) \leq v_1(1,0,0) + v_1(0,1,0) - v_1(0,0,0)$$
We continue to construct all the values iteratively in the same way. 

We ran an extended algorithm for over $2,000,000$ different random valuations.
For each run of the algorithm, we checked if the social welfare approximation of 2 was achieved. All the runs successfully passed the approximation test. The code for the simulations is publicly available.%
\footnote{See \url{https://colab.research.google.com/drive/1vkhOt3aMG5DgivaHZt5URNhPJ7q0W4J1}.}

The algorithm that we ran is very similar to the one shown in this paper for binary signals but with one change. The priorities are defined differently. Given that the range of possible signals is $[k]$; we ordered the priorities from $k$ to $1$. For each signals profile $s$, we iterate over the signal range from $k$ to $1$ and for the iteration of signal's value $m$  we check for priorities $1$ and $2$ as follows; {\bf Priority 1} at iteration $m$ consists of two bidders $i$ and $j$ s.t.:
\begin{enumerate}
    \item Bidder $i$ and $j$'s signals are at least $m$; 
    \item Bidders $i$ and $j$ can both be colored {red} at iteration $s$;
    \item No other bidder $k\neq i,j$ is colored {red} at the beginning of iteration $s$;
    \item The sum of values $v_i(s)+v_j(s)$ is at least $\OPT(s)$ (recall that $\OPT(s)$ is the highest value of any bidder for the item given signal profile $s$).
\end{enumerate}
{\bf Priority 2} at iteration $m$ consists of one bidder $i$ s.t.:

\begin{enumerate}
    \item Bidder $i$'s signal is at least $m-1$;
    \item Bidder $i$ can be colored {red} at iteration $s$;
    \item The value $v_i(s)$ equals $\OPT(s)$.
\end{enumerate}
In this extended algorithm we get $2k$ priorities, $2$ for each possible value of a signal. We incorporate the sense of high and low signals by ordering the priorities decreasingly by the signal range.
\subsection{Extending Lemma~\ref{lemma2}}

We show that our main technical lemma, Lemma~\ref{lemma2}, extends beyond binary signals.

\begin{lemma}
\label{lemma-int-signals}
Consider a signals profile $s'=(s_1',...,s_n')$ and $3+\ell_1+\ell_2+\ell_3$ bidders\\ $E=\{i,j,k$, $t_1,\dots,t_{\ell_1}$,$t_1',\dots,t_{\ell_2}'$,$t''_1,\dots,t''_{\ell_3}\}$ (not necessarily distinct) with SOS valuations, and $n$ values $d_1,d_2,...,d_n$. Define $s_h=s_h'+d_h$ for each bidder $h\in [n]$ and $s=(s_1,...,s_n)$. If the following $3+\ell_1+\ell_2+\ell_3$ inequalities hold simultaneously then they hold with equality:
\begin{itemize}
    \item $v_i(s_i',s_{-i})\geq v_j(s_i',s_{-i})+v_{t_1}(s_i',s_{-i})$
    
    \item $\forall 0<h\leq\ell_1$: $v_{t_h}(s_i',s_{t_1}',...,s_{t_h}',s_{-\{i,t_1,...,t_h\}})\geq \\ v_j(s_i',s_{t_1}',...,s_{t_h}',s_{-\{i,t_1,...,t_h\}})+v_{t_{h+1}}(s_i',s_{t_1}',...,s_{t_h}',s_{-\{i,t_1,...,t_h\}})$
    
    \item $v_j(s_j',s_{-j})\geq v_i(s_j',s_{-j})+u_{t_1'}(s_j',s_{-j})$
    
    \item $\forall 0<h\leq\ell_2$: $v_{t_h'}(s_j',s_{t_1'}',...,s_{t_h'}',s_{-\{j,t_1',...,t_h'\}}) \geq \\ v_i(s_j',s_{t_1'}',...,s_{t_h'}',s_{-\{j,t_1',...,t_h'\}})+v_{t_{h+1}'}(s_j',s_{t_1'}',...,s_{t_h'}',s_{-\{j,t_1',...,t_h'\}})$
    
    \item $\forall 0<h\leq\ell_3$: $v_{t''_h}(s_k',s_{t''_1}',...,s_{t''_h}',s_{-\{k,t''_1,...,t''_h\}}) \geq \\ v_j(s_k',s_{t''_1}',...,s_{t''_h}',s_{-\{k,t''_1,...,t''_h\}})+v_{t''_{h+1}}(s_k',s_{t''_1}',...,s_{t''_h}',s_{-\{k,t''_1,...,t''_h\}})$
    
    \item $v_k(s_k',s_{-k})\geq v_{t''_1}(s_k',s_{-k})+v_j(s_k',s_{-k})$
\end{itemize}

\end{lemma}

\begin{proof}
For each profile of signals $\widetilde{s}$ define the corresponding set of signals $\widetilde{S}=\{h\mid  \widetilde{s}_h=s_h=s_h'+d_h\}$.
$v(\widetilde{S})$ is a submodular set function, and thus, defining $u_i(x)=v_i(x\mid S')$, we can rewrite the inequalities as the following:

\begin{itemize}
    \item $u_i(jkt_1t_2...t_{\ell_1}t_1't_2'...t_{\ell_2}'t''_1,...,t''_{\ell_3}) +v_i(S')\geq \\ u_j(jkt_1t_2...t_{\ell_1}t_1't_2'...t_{\ell_2}'t''_1,...,t''_{\ell_3})+v_j(S')+\\u_{t_1}(jkt_1t_2...t_{\ell_1}t_1't_2'...t_{\ell_2}'t''_1,...,t''_{\ell_3})+v_{t_1}(S')$
    
    \item $\forall 0<h\leq\ell_1$: $u_{t_h}(jkt_{h+1}...t_{\ell_1}t_1't_2'...t_{\ell_2}'t''_1,...,t''_{\ell_3}) +v_{t_h}(S')\geq \\ u_j(jkt_{h+1}...t_{\ell_1}t_1't_2'...t_{\ell_2}'t''_1,...,t''_{\ell_3})+v_j(S')+\\u_{t_{h+1}}(jkt_{h+1}...t_{\ell_1}t_1't_2'...t_{\ell_2}'t''_1,...,t''_{\ell_3})+v_{t_{h+1}}(S')$
    
    \item $u_j(ikt_1t_2...t_{\ell_1}t_1't_2'...t_{\ell_2}'t''_1,...,t''_{\ell_3}) +v_j(S')\geq \\ u_i(ikt_1t_2...t_{\ell_1}t_1't_2'...t_{\ell_2}'t''_1,...,t''_{\ell_3})+v_i(S')+\\+u_{t_1'}(ikt_1t_2...t_{\ell_1}t_1't_2'...t_{\ell_2}'t''_1,...,t''_{\ell_3})+v_{t_1'}(S')$
    
    \item $\forall 0<h\leq\ell_2$: $u_{t_h'}(ikt_1...t_{\ell_1}t_{h+1}'...t_{\ell_2}'t''_1,...,t''_{\ell_3}) +v_{t_h'}(S')\geq \\ u_i(ikt_1...t_{\ell_1}t_{h+1}'...t_{\ell_2}'t''_1,...,t''_{\ell_3})+v_i(S')+\\u_{t_{h+1}'}(ikt_1...t_{\ell_1}t_{h+1}'...t_{\ell_2}'t''_1,...,t''_{\ell_3})+v_{t_{h+1}'}(S')$
    
    \item $\forall 0<h\leq\ell_3$: $u_{t''_h}(ijt_1t_2...t_{\ell_1}t_1't_2'...t_{\ell_2}'t''_{h+1},...,t''_{\ell_3}) +v_{t''_h}(S')\geq \\ u_j(ijt_1t_2...t_{\ell_1}t_1't_2'...t_{\ell_2}'t''_{h+1},...,t''_{\ell_3})+v_j(S')+\\u_{t''_{h+1}}(ijt_1t_2...t_{\ell_1}t_1't_2'...t_{\ell_2}'t''_{h+1},...,t''_{\ell_3})+v_{t''_{h+1}}(S')$
    
    \item $u_k(ijt_1t_2...t_{\ell_1}t_1't_2'...t_{\ell_2}'t''_1,...,t''_{\ell_3})+v_k(S')\geq \\ u_{t''_1}(ijt_1t_2...t_{\ell_1}t_1't_2'...t_{\ell_2}'t''_1,...,t''_{\ell_3})+v_{t''_1}(S')+\\+u_j(ijt_1t_2...t_{\ell_1}t_1't_2'...t_{\ell_2}'t''_1,...,t''_{\ell_3})+v_j(S')$
\end{itemize}
By Lemma~\ref{lemma2} we know that if these inequalities hold together, they hold with equality.
\qed
\end{proof}

\end{document}